\documentclass[10pt,a4paper]{article}

\usepackage{amsmath,amsthm,amsfonts,amssymb}
\usepackage{hyperref}
\usepackage{graphicx}
\usepackage{enumerate}
\usepackage{eucal}
\newcommand{\ie}{\emph{i.e.}}
\newcommand{\eg}{\emph{e.g.}}
\newcommand{\cf}{\emph{cf}}
\newcommand{\sii}{L^2}
\newcommand{\Dom}{\mathfrak{D}}
\newcommand{\diag}{\mathop{\mathrm{diag}}\nolimits}
\newcommand{\sgn}{\mathop{\mathrm{sgn}}\nolimits}
\newcommand{\Real}{\mathbb{R}}
\newcommand{\Com}{\mathbb{C}}
\newcommand{\Nat}{\mathbb{N}}
\newcommand{\PT}{\mathcal{PT}}
\renewcommand{\P}{\mathcal{P}}
\newcommand{\T}{\mathcal{T}}
\newcommand{\ii}{{\rm i}}
\newcommand{\dd}{{{\rm d}}}
\newtheorem{Lemma}{Lemma}
\newtheorem{Proposition}{Proposition}
\theoremstyle{definition}
\newtheorem{Remark}{Remark}
\numberwithin{equation}{section}
\begin{document}
%
%
\title{\bf The Pauli equation with complex boundary conditions}
\author{D.~Kochan,$^{a,b}$ \ D.~Krej\v{c}i\v{r}\'{i}k,$^{c}$ \
R.~Nov\'{a}k\,$^{c,d}$ \ and \ P.~Siegl\,$^{c,e	}$}
\date{
\small
\emph{
\begin{quote}
\begin{itemize}
\item[$a)$] 
Institute for Theoretical Physics, University of Regensburg, 
Universit\"atsstrasse 31, 93053 Regensburg, Germany 
\item[$b)$]
Department of Theoretical Physics, FMFI UK, 
Mlynsk\'a dolina F2, 842 48 Bratislava, Slovakia
\item[$c)$] 
Department of Theoretical Physics, 
Nuclear Physics Institute ASCR, 25068 \v{R}e\v{z}, Czech Republic
\item[$d)$] 
Faculty of Nuclear Sciences and Physical Engineering, 
Czech Technical University in Prague, 
B\v{r}ehov\'a 7, 11519 Praha 1, Czech Republic
\item[$e)$] 
Group of Mathematical Physics of the University of Lisbon,
Complexo Interdisciplinar,
Av.~Prof.~Gama Pinto~2, P-1649-003 Lisboa, Portugal
\item[E-mails:]
denis.kochan@ur.de,
krejcirik@ujf.cas.cz, 
novakra9@fjfi.cvut.cz, 
siegl@ujf.cas.cz
\end{itemize}
\end{quote}
}
\medskip
22 March 2012
}
\maketitle
\begin{abstract}
\noindent
We consider one-dimensional Pauli Hamiltonians
in a bounded interval with possibly non-self-adjoint 
Robin-type boundary conditions.
We study the influence of the spin-magnetic interaction 
on the interplay between the type of boundary conditions
and the spectrum.
A special attention is paid to $\PT$-symmetric
boundary conditions with the physical choice
of the time-reversal operator~$\T$.
\medskip
\begin{itemize}
\item[\textbf{MSC\,2010:}]
Primary: 34L40, 34B08; 81Q12;
Secondary: 34B07, 81V10 
\item[\textbf{Keywords:}]
Pauli equation, spin-magnetic interaction, 
Robin boundary conditions,
non-self-adjoint\-ness, non-Hermitian quantum mechanics,
$\PT$-symmetry, time-re\-ver\-sal operator
\end{itemize}
\end{abstract}
%

\vfill
{\small
}

\newpage
\section{Introduction}
%
In recent years there has been a growing interest  
in non-Hermitian ``extensions'' of quantum mechanics,
usually associated with the names of \emph{$\mathcal{PT}$-symmetry},
\emph{pseudo-Hermiticity}, \emph{quasi-Hermiticity} 
or \emph{crypto-Hermiticity}
(we respectively refer to \cite{Bender_2007,Ali-review,GHS,Znojil_2008}
where the first two works are recent surveys with many references).
The quotation marks are used here because
the extended theories are physically relevant only if 
the operators in question are 
\emph{similar to self-adjoint} operators,
which in turn puts the concept back to the conventional
quantum mechanics.

However, the freedom related to the existence of the similarity
transformation can be highly useful in applications, 
since a complicated non-local self-adjoint
operator can be represented by a (possibly non-self-adjoint) 
differential operator 
(see~\cite{KSZ} for one-dimensional examples), 
and the spectral theory for the latter is much more developed.
Moreover, it is necessary that the non-Hermitian operators
possess \emph{real spectra}, which can be often ensured 
(at least in some perturbative regimes \cite{CGS,Langer-Tretter_2004})
by the simple criterion of $\mathcal{PT}$-symmetry.

The goal of the present paper is to examine 
the role of \emph{spin} in the above theories.  
We consider the simplest non-trivial situation
of an electron (spin~$\frac{1}{2}$, mass~$m$, charge~$-e<0$)
interacting exclusively 
with an external homogeneous magnetic field $\vec{B}\in\Real^3$.
Choosing the Poincar\'e gauge in which the magnetic vector potential 
coincides with $\frac{1}{2} \vec{B} \times \vec{x}$,
this system is governed by the \emph{Pauli equation}
\begin{equation}\label{Pauli}
  i \hbar \frac{\partial\Psi}{\partial t}
  = -\frac{\hbar^2}{2 m} \, \Delta\Psi
  + \frac{\mu}{\hbar} \, \vec{B}\cdot\vec{L} \, \Psi
  + \frac{e^2}{8 m} \, (\vec{B} \times \vec{x})^2 \Psi
  + \mu \, \vec{B}\cdot\vec{\sigma} \, \Psi
  =: H\Psi
\end{equation}
in the space-time variables $(\vec{x},t)$,
where~$\hbar$ is the reduced Planck constant,
$\mu:=\hbar e/(2m)$ is the Bohr magneton (for simplicity),
$\vec{L}$ is the angular-momentum operator
and~$\vec{\sigma}$ is a three-component vector 
formed by the Pauli matrices. 
The spinorial wavefunction~$\Psi$ can be represented
as an element of $\sii(\Real^3)\otimes\Com^2$ 
and the operators appearing in~\eqref{Pauli}
are assumed to appropriately act in this Hilbert space.

The Hamiltonian~$H$ 
(equipped with a suitable domain)
is Hermitian when considered in the full Hilbert space
$\sii(\Real^3)\otimes\Com^2$.
Moreover, the Pauli equation~\eqref{Pauli}
is invariant under a simultaneous
reversal of the space and time variables
(\cf~the discussion in Section~\ref{Sec.symmetry}).
Relying on general definitions for the Dirac field 
(see, \eg, \cite[\textsection26]{Landau-Lifschitz-4})
and the fact that the Pauli equation can be obtained from
the Dirac equation in a non-relativistic limit,
the discrete symmetries can be represented by means of 
the \emph{parity}~$\mathcal{P}$ 
and the \emph{time-reversal operator}~$\mathcal{T}$
(uniquely determined up to a phase factor).
 
Our way how to ``complexify''~\eqref{Pauli} 
is to restrict the space variables 
to a subset $\Omega\subset\Real^3$ 
and impose \emph{complex boundary conditions}
of the Robin type
\begin{equation}\label{bc}
  \frac{\partial\Psi}{\partial n} + A \Psi = 0 
  \qquad\mbox{on}\qquad
  \partial\Omega
  \,,
\end{equation}
where~$n$ is the outward pointing normal unit to~$\partial\Omega$
and~$A$ is a two-by-two complex-valued matrix. 
If~$\Omega$ is invariant with respect 
to the spatial inversion~$\mathcal{P}$,
it is possible to choose~$A$ in such a way that
the $\mathcal{PT}$-symmetry of~\eqref{Pauli}
remains valid for the (possibly non-Hermitian) operator~$H$
on $\sii(\Omega)\otimes\Com^2$,
subject to the boundary conditions~\eqref{bc}.

In this paper we study the interplay between the form
of the matrix~$A$ and the spectrum of~$H$.
In particular, we are interested in the existence
of real eigenvalues 
in the $\mathcal{PT}$-symmetric situation. 

We are not aware of previous works on Pauli equation 
in the non-Hermitian extensions of quantum mechanics.
However, there exist results on 
spinorial systems in the context of 
$\mathcal{PT}$-symmetric coupled-channels models
\cite{Znojil_2006a,Znojil_2006c,Znojil_2006b}
and the Dirac equation in the framework of Krein spaces
\cite{Albeverio-Guenther-Kuzhel_2009,Kuzhel-Trunk_2011}.

One of the reasons for considering the spinorial model in this paper 
is the fact that the time-reversal operator~$\mathcal{T}$ 
\emph{differs from the complex conjugation},
the latter being the time-reversal operator
for the scalar (\ie~spinless) Schr\"odinger equation,
widely studied in the $\mathcal{PT}$-symmetric quantum theory.
In fact, for fermionic systems (\ie~half-integer non-zero spin), 
one has
\begin{equation}\label{antiunitary}
  \mathcal{T}^2 = -1
  \,.
\end{equation}
This has been remarked previously in the context of 
pseudo-Hermitian operators in 
\cite{Scolarici-Solombrino_2002,Blasi-Scolarici-Solombrino_2004}.
A generalized concept of $\mathcal{PT}$-sym\-met\-ry
as regards the operator $\mathcal{P}$
is suggested in \cite{Znojil_2011}.

The present model can be regarded as an extension
of the one-dimensional scalar Hamiltonians with complex Robin 
boundary conditions studied in \cite{KBZ,K4,KSZ} to the spinorial case. 
We refer to~\cite{KS,HKS} for the discussion of relevance
of (possibly non-Hermitian) Robin boundary conditions in physics
and, in particular, 
to Section~\ref{Sec.scattering} for a simple scattering-type
interpretation in the present setting.

This paper is organized as follows.
In the following section we specify our model 
in terms of a one-dimensional Hamiltonian
coming from~\eqref{Pauli}.
A physical relevance of the boundary conditions~\eqref{bc}   
is suggested in Section~\ref{Sec.scattering}.
Section~\ref{Sec.Hamiltonian.forms} is devoted to
a rigorous definition of our Hamiltonian
as a closed operator associated with
a sectorial sesquilinear form.
In Section~\ref{Sec.symmetry} we discuss
the physical choice of the operator~$\PT$
and establish conditions on the boundary matrix~$A$
which guarantee various symmetry properties 
of the Hamiltonian.	
Section~\ref{Sec.spectrum} is devoted to 
a spectral analysis supported by numerics;
on several $\PT$-symmetric examples 
we discuss the dependence of the spectrum
on parameters characterizing the matrix~$A$.
The paper is concluded by Section~\ref{Sec.end}
in which we mention some open problems.

\section{Our model}\label{Sec.Hamiltonian}
%
We begin specifying our model represented 
by the Pauli equation~\eqref{Pauli}.

We choose the coordinate system in~$\Real^3$ in such a way
that the third coordinate axis is parallel with 
the homogeneous magnetic field~$\vec{B}$,
\ie~$\vec{B}=(0,0,B)$ where~$B\in\Real$. 
Then the orbital interaction $\vec{B}\cdot\vec{L}$ 
and the diamagnetic term $(\vec{B}\times\vec{x})^2$
represent differential operators in the first two space variables only.
On the other hand, the spinorial interaction $\vec{B}\cdot\vec{\sigma}$
acts in the third space variable only
(
through the Pauli matrix $\sigma_3=\diag(1,-1)$). 

We set 
\begin{equation}\label{domain.choice}
  \Omega:=\Real^2\times(-a,a)
  \,,
\end{equation}
with some positive number~$a$.
Assuming that the matrix~$A$ in~\eqref{bc} is constant 
on each of the connected components of~$\partial\Omega$,
the spectral problem for the Hamiltonian~$H$ therefore splits 
into two separate problems: a two-dimensional Landau-level
problem in the first two variables and 
a one-dimensional problem in the third variable
which we will study in sequel.
Up to a constant factor representing 
the energy of the given Landau level,
the corresponding one-dimensional operators have the form   
\begin{equation}\label{fyzham}
  H_b = 
\begin{pmatrix}
  - \frac{\displaystyle\dd ^2}{\displaystyle\dd x^2} + b & 0\\
  0 & - \frac{\displaystyle\dd ^2}{\displaystyle\dd x^2} - b
\end{pmatrix}
  \qquad\mbox{on}\qquad
  \mathcal{H} := \sii\big((-a,a);\Com^2\big)
  \,,
\end{equation}
subject to the boundary conditions 
\begin{equation}\label{bc.pm}
  \Psi'(\pm a) + A^\pm \Psi(\pm a) = 0 
  \,.
\end{equation}
Here we have put $\hbar^2/(2m)=1$ and $b:=\mu B$,
$A^\pm \in \Com^{2 \times 2}$, and the third space
variable is (with an abuse of notation) denoted by~$x$.

In view of the choice of physical constants made above,
the only distinguished length in our problem is the half-width~$a$,
and therefore the results must be scaled appropriately
with respect to this length.
In particular, the parameter~$b$ (characterizing the strength
of the magnetic field) and eigenvalues of~$H_b$ 
(corresponding to quantum energies)
become dimensionless when multiplied by~$a^{2}$.
The same can be done for the entries of~$A^{\pm}$ 
when multiplied by~$a$.
Consequently, all parameters can be thought as dimensionless
in the sequel.

As usual, the Hilbert space~$\mathcal{H}$ is identified
with $\sii((-a,a)) \otimes \Com^2$ and its elements 
are represented by the two-component spinors
$$
  \Psi = 
  \begin{pmatrix} 

  \psi_+ \\ \psi_-
  \end{pmatrix}
  \,,
$$
where $\psi_\pm \in \sii((-a,a))$
(the $\pm$ notation should not be confused with 
the superscripts of the matrices~$A^\pm$
referring to the endpoints of $(-a,a)$).
The inner product in~$\mathcal{H}$ is defined by
$$
  (\Phi, \Psi) 
  := \int_{-a}^a \overline{\Phi}^T\!(x) \, \Psi(x) \, {\rm d}x
  \,,
$$ 
where the upper index $T$ denotes transposition. 
The corresponding norm is denoted by~$\|\cdot\|$. 
The Euclidean norm of the spinor~$\Psi$ 
as a vector in~$\Com^2$ is denoted by
$
  |\Psi| := \sqrt{|\psi_+|^2+|\psi_-|^2}
$
and we use the same notation for the corresponding
operator (matrix) norm 
$
  |A| := \max\big\{
  |A\Psi| \, \big| \, \Psi\in \Com^2, \, |\Psi|=1 
  \big\} 
$
for $A \in \Com^{2 \times 2}$.

\section{A scattering motivation}\label{Sec.scattering}
%
Before giving a rigorous definition of our Hamiltonian
formally introduced \eqref{fyzham}--\eqref{bc.pm},
let us first justify the physical relevance of 
the boundary conditions~\eqref{bc.pm}. 
Our method is based on a generalization of an idea 
originally suggested in~\cite{HKS}. 

Consider a generalized eigenvalue problem for the Hamiltonian
of the form \eqref{fyzham} on the whole space~$\Real$
locally perturbed by an electric field:
\begin{equation}\label{fyzham.scattering}
\begin{pmatrix}
  - \frac{\displaystyle\dd ^2}{\displaystyle\dd x^2} + b + V(x) & 0\\
  0 & - \frac{\displaystyle\dd ^2}{\displaystyle\dd x^2} - b + V(x)
\end{pmatrix}
\begin{pmatrix}
  \psi_+ \\
  \psi_-
\end{pmatrix}
  =
  \lambda
\begin{pmatrix}
  \psi_+ \\
  \psi_-
\end{pmatrix}
  \,.
\end{equation}
Here~$x\in\Real$, $\lambda\in\Real$ and 
$V$~is the electric potential that is assumed to be
compactly supported in $(-a,a)$. 
Solutions~$\Psi$ with $\lambda<-|b|$ 
are bound states (associated with discrete eigenvalues),
while those with $\lambda \geq -|b|$ correspond to scattering states
(associated with the essential spectrum).	

Outside the support of~$V$ the problem~\eqref{fyzham.scattering} 
admits explicit solutions in terms of exponential functions.
Let us look for special scattering solutions satisfying
\begin{equation}\label{as.sol}
  \Psi(x) = 
\begin{pmatrix}
  e^{\ii \sqrt{\lambda - b} \, x}
  \\
  e^{\ii \sqrt{\lambda + b} \, x}
\end{pmatrix}
  \qquad \mbox{for} \qquad
  |x| \geq a
  \,.
\end{equation}
Then the (physical) problem~\eqref{fyzham.scattering} 
on the whole real axis can be solved by considering
an (effective) boundary value problem in $(-a,a)$.
The latter is simply obtained by considering~\eqref{fyzham.scattering}  
in $(-a,a)$ and requiring that the solutions  
match at~$\pm a$ smoothly with the asymptotic solutions~\eqref{as.sol}.    
This leads to the boundary conditions~\eqref{bc.pm}
with an \emph{energy-dependent} matrix
\begin{equation}\label{bc.energy}
  A_\lambda^\pm =
\begin{pmatrix}
  -\ii \sqrt{\lambda - b} & 0
  \\
  0 & -\ii \sqrt{\lambda + b}
\end{pmatrix}
  \,.
\end{equation}

Note that~\eqref{fyzham.scattering} for $x \in (-a,a)$, 
subject to~\eqref{bc.pm} with~\eqref{bc.energy} at~$\pm a$,
does not represent a standard spectral problem,
it is rather an operator-pencil problem
(because of the dependence of~$A_\lambda^\pm$
on the spectral parameter~$\lambda$).  
It is non-linear in its nature.
However, it can be solved by first considering 
a genuine (linear) spectral problem,
namely~\eqref{fyzham.scattering} for $x \in (-a,a)$, 
subject to~\eqref{bc.pm} with~$A_\alpha^\pm$ at~$\pm a$,
with~$\alpha$ being treated 	as a real parameter. 
This leads to a discrete set of eigencurves 
$\alpha \mapsto \lambda_n(\alpha)$, $n \in \Nat$.
Then the ``eigenvalues'' of the true, energy-dependent problem   
are determined as those points~$\lambda_n(\alpha_*)$ 
satisfying the (non-linear) algebraic equations
\begin{equation}\label{parabola}
  \lambda_n(\alpha_*) = \alpha_*  
  \,.
\end{equation}

The elements of the set $\{\lambda_n(\alpha_*)\}_{n\in\Nat}$ 
are called \emph{perfect-transmission energies} (PTEs) in~\cite{HKS},
since their physical meaning is that they determine energies 
for which there is no reflection for the initial 
scattering problem~\eqref{fyzham.scattering} in~$\Real$. 
It is interesting that PTEs are real, although they are 
obtained via solving a highly non-self-adjoint spectral problem. 
This feature is related to the fact that the choice~$A_\alpha^\pm$
ensures that the boundary conditions are $\mathcal{PT}$-symmetric
(although not $\mathcal{PT}$-symmetric in the context
of the present paper where we do not allow the presence
of~$\lambda$ and~$b$ in the boundary conditions, see below).
A physical interpretation of the possible complexification
of the spectra of the auxiliar $\mathcal{PT}$-symmetric 
spectral problem is also proposed in~\cite{HKS}.

It is also interesting to note that switching on
the static magnetic field (\ie~making $b \not= 0$)
will typically lead to a splitting of the doubly degenerate eigencurves 
corresponding to the auxiliar non-self-adjoint spectral problem for $b=0$
(\cf~Figure~\ref{exampleCb}).
Consequently, to each of the PTE in the scalar case 
without the magnetic field there correspond two PTEs
in our spinorial model.
The analogy with the \emph{Zeeman effect} 
should not be surprising.

The matrix~\eqref{bc.energy} is complex and non-Hermitian,
which is typical for effective models of scattering solutions
of~\eqref{fyzham.scattering}. 
On the other hand, real-valued Hermitian matrices 
are obtained when looking for bound states.
In this paper, we proceed in a full generality
by allowing \emph{arbitrary} matrices~$A^\pm$ in~\eqref{bc.pm}.  
However, it is important to stress that 
we regard the matrices
as parameters entering the spectral system;
the dependence of~$A^\pm$ on the spectral parameter~$\lambda$
is not allowed and the dependence on the field~$b$
is allowed only if~$b$ is treated as a parameter 
(no change under the action of~$\T$, 
\cf~Section~\ref{Sec.symmetry}).  

To end up this motivation section, 
let us note that alternative proposals for 
the connection between non-Hermitian 
$\mathcal{PT}$-sym\-met\-ric operators
and physics have been suggested recently in the context of scattering in
\cite{Jin-Song_2011,Longhi_2011,Yoo-Sim-Schomerus_2011,Jones_2011}.

\section{The Pauli Hamiltonian}\label{Sec.Hamiltonian.forms}
%
We now turn to a rigorous definition of the Hamiltonian
formally introduced by \eqref{fyzham}--\eqref{bc.pm}.
In other words, since we are interested in spectral properties,
we need a closed realization of the operator~$H_b$.

The easiest way is to define the Hamiltonian as 
the Friedrichs extension of the operator~\eqref{fyzham}
initially considered on uniformly smooth spinors 
satisfying~\eqref{bc.pm}.
On such a restricted domain, 
an integration by parts easily leads to 
the associated sesquilinear form~$h_b$ 
as a sum of three terms
\begin{equation}\label{form}
  h_b(\Phi,\Psi) 
  = q_1(\Phi,\Psi) + b \, q_2(\Phi,\Psi) + q_3(\Phi,\Psi) 
  \,,
  \\
\end{equation}
where
\begin{equation}\label{form3}
\begin{aligned}
  q_1(\Phi,\Psi) &:= (\Phi',\Psi') 
  \,,
  \\
  q_2(\Phi,\Psi) &:= (\Phi,\sigma_3\Psi) 
  \,, 
  \\
  q_3(\Phi,\Psi) &:=
  \overline{\Phi}^T\!(a)\,A^+\,\Psi(a) 
  - \overline{\Phi}^T\!(-a)\,A^-\,\Psi(-a) 
  \,. 
\end{aligned}
\end{equation}
The form~$h_b$ is well defined on a larger,
Sobolev-type space
\begin{equation}\label{form.domain}
  \Dom(h_b) := H^1\big((-a,a);\Com^2\big)
  \,.
\end{equation}
It is obvious for~$q_1$ and~$q_2$, 
while the boundary term~$q_3$ can be shown bounded
on~$\Dom(h_b)$ by means of the Sobolev embedding
$
  H^1((-a,a)) \hookrightarrow C^0([-a,a])
$.

Our aim is to show that~$h_b$ is a closed sectorial form.
It is clear for~$q_1$ defined on~\eqref{form.domain},
since~$q_1$ is associated with the Neumann Laplacian
(\cf~\cite[Sec.~7]{Davies}),
and as such it is a densely defined, closed, 
symmetric, non-negative form. 
The term~$q_2$ represents just a bounded perturbation;
indeed, $|q_2[\Psi]| \leq \|\Psi\|^2$
for every $\Psi \in \mathcal{H}$.
It is not longer true for~$q_3$, however,
a suitable quantification of the Sobolev embedding 
can be used to ensure that~$q_3$ still represents 
a small perturbation in the following sense.
\begin{Lemma}\label{Lem.relative}
For every $\Psi\in\Dom(h_b)$ and $\varepsilon\in(0,1)$,
\begin{equation*}
\begin{aligned}
  \big|q_3[\Psi]\big| 
  &\leq \varepsilon \left(|A^+| + |A^-|\right) \|\Psi'\|^2 
  + \left(
  \frac{|A^+| + |A^-|}{2a} + \frac{|A^+| + |A^-|}{\varepsilon}
  \right)
  \|\Psi\|^2
  \,.
\end{aligned}
\end{equation*}
%
Consequently, 
the form $q_2+q_3$ is relatively bounded with respect to~$q_1$  
and the relative bound can be made arbitrarily small.
\end{Lemma}
\begin{proof}
The claim is based on the estimates
\begin{equation}\label{fundamental}
  |\Psi(\pm a)|^2 
  \leq 2 \;\!\|\Psi'\| \|\Psi\| + \frac{1}{2a}\,\|\Psi\|^2 
  \leq \varepsilon \, \|\Psi'\|^2 
  + \left(
  \frac{1}{2a} + \frac{1}{\varepsilon}
  \right)
  \|\Psi\|^2 
\end{equation}
valid for any $\Psi\in\Dom(h_b)$.
Here~the first inequality can be established quite easily
by the fundamental theorem of calculus 
and the Schwarz inequality. 
\end{proof}

Consequently, the perturbation result \cite[Thm.VI.1.33]{Kato} 
can be used to show that~$h_b$ is indeed sectorial and closed. 
According to the first representation theorem \cite[Thm.VI.2.1]{Kato}, 
there exists a unique m-sectorial operator~$H_b$ in~$\mathcal{H}$ 
such that $h_b(\Phi,\Psi) = (\Phi,H_b \Psi)$ 
for all $\Phi \in \Dom(h_b)$ 
and $\Psi \in \Dom(H_b) \subset \Dom(h_b)$. 
Following the arguments \cite[Ex.~VI.2.16]{Kato}, 
it is easy to check that~$H_b$ indeed acts 
as \eqref{fyzham}--\eqref{bc.pm}; more precisely,
\begin{equation}\label{hamiltonian}
\begin{aligned}
  H_b \Psi &=
\begin{pmatrix}
  -\psi_+'' + b \psi_+ \\
  -\psi_-'' - b \psi_-
\end{pmatrix}, \\
  \Dom(H_b) &= \Big\{ \Psi \in H^2\big((-a,a);\mathbb{C}^2\big) 
  \ \big| \
  \Psi'(\pm a) + A^\pm \Psi(\pm a) = 0 
  \Big\}
  \,.
\end{aligned}
\end{equation}
\begin{Proposition}\label{msectorial}
$H_b$~defined by~\eqref{hamiltonian} is an m-sectorial  
operator on~$\mathcal{H}$.
The adjoint of~$H_b$ is given by
\begin{equation}\label{hamiltonian.adjoint}
\begin{aligned}
  H_b^* \Psi &=
\begin{pmatrix}
  -\psi_+'' + b \psi_+ \\
  -\psi_-'' - b \psi_-
\end{pmatrix}, \\
  \Dom(H_b^*) &= \Big\{ \Psi \in H^2\big((-a,a);\mathbb{C}^2\big) 
  \ \big| \
  \Psi'(\pm a) + (A^\pm)^* \;\! \Psi(\pm a) = 0 
  \Big\}
  \,,
\end{aligned}
\end{equation}
where $A^*=\overline{A^T}$.
\end{Proposition}
\begin{proof}
It remains to notice (\cf~\cite[Thm.~VI.2.5]{Kato})
that the adjoint operator
is determined as the m-sectorial operator 
associated with the adjoint form~$h_b^*$
defined by $h_b^*(\Phi,\Psi):=\overline{h_b(\Psi,\Phi)}$,
$\Dom(h_b^*):=\Dom(h_b)$.
\end{proof}

Note that the choice $A^\pm=0$ gives rise 
to the (self-adjoint) Pauli Hamiltonian,
subject to Neumann boundary conditions,
that we denote by~$H_b^N$.  
(At the same time, the choice $A^\pm=\infty$
formally corresponds to Dirichlet 
boundary conditions.)

\section{Symmetry properties}\label{Sec.symmetry}
%
It is well known that the Pauli equation~\eqref{Pauli}
(in the whole space $\Real^3$)
is invariant under the simultaneous space inversion and time reversal
(\ie~$\vec{x} \mapsto -\vec{x}$ and $t \mapsto -t$, respectively).
This can be easily established if one realizes
that the time reversal leads to a change of orientation
of the magnetic field (\ie~$\vec{B} \mapsto -\vec{B}$),
while the orientation is unchanged by the space inversion.
These properties can be deduced from Maxwell's equations
to which the equation~\eqref{Pauli} is implicitly coupled
(\cf~\cite[\textsection17]{Landau-Lifschitz-2}).

One is tempted to mathematically formalize
the space-time reversal invariance 
in terms of a symmetry property of the Hamiltonian~$H$.
Given a unitary or antiunitary operator~$\mathcal{C}$,
we say that a linear operator~$H$ in a Hilbert space
is \emph{$\mathcal{C}$-symmetric} if 
\begin{equation}\label{invariance}
  [H,\mathcal{C}] = 0
  \,.
\end{equation}
Here the commutator relation should be interpreted 
as an operator identity on the domain of~$H$,
\ie~$\mathcal{C} H \subset H \mathcal{C}$.
In this framework, however, the Hamiltonian~$H$
appearing~\eqref{Pauli} is not $\mathcal{PT}$-symmetric,
just because there is no way how ensure
the change of sign~$\vec{B}$ under the action of~$\T$
in the Hilbert-space setting
(in which $\vec{B}$ is considered as an operator of multiplication).
Nevertheless, $H$~of course satisfies~\eqref{invariance}
with $\mathcal{C}=\mathcal{PT}$ provided that
the magnetic field is absent.

One of the goals of this section is to determine
the class of boundary matrices~$A^\pm$ which
preserves the $\mathcal{PT}$-symmetry in the sense above.
In other words, since we do not like to think of~$b$ as a component
of a field governed by the additional equations
and to mathematically formalize the action of~$\mathcal{T}$ on the field
($b$~is rather a fixed parameter in our Hilbert-space setting),
we restrict ourselves to rigorously looking for the property
\begin{equation}\label{invariance.1D}
  [H_0,\mathcal{PT}]=0
  \,;
\end{equation}
boundary conditions~\eqref{bc.pm} satisfying this relation
will be called $\mathcal{PT}$-symmetric.
In other words, boundary conditions are $\mathcal{PT}$-symmetric
if, and only if, $\PT\Psi$ satisfies the same equations
as~$\Psi$ in~\eqref{bc.pm}.
Similarly, we shall define $\mathcal{PK}$-symmetric
boundary conditions.

In our one-dimensional situation~\eqref{fyzham},
the parity~$\mathcal{P}$ 
and the time reversal operator~$\mathcal{T}$  
act on spinors as follows
(\cf~\cite[\textsection30]{Landau-Lifschitz-3} and
\cite[\textsection60]{Landau-Lifschitz-3}, respectively)
\begin{equation}\label{PT-operators}
  (\mathcal{P}\Psi)(x)
  := \Psi(-x)
  \,, \qquad
  (\mathcal{T}\Psi)(x)
  := i\sigma_2\overline{\Psi(x)} =
  \begin{pmatrix}
  \overline{\psi_-(x)} \\ -\overline{\psi_+(x)}
  \end{pmatrix}
  \,.
\end{equation}
It is important to stress that~$\mathcal{T}$
differs from the complex conjugation operator
\begin{equation}\label{K-operator}
  (\mathcal{K}\Psi)(x):=\overline{\Psi(x)}
  \,,
\end{equation}
the latter being the time reversal operator in the scalar case.

It is easily seen that $\mathcal{P}$, $\mathcal{T}$ and $\mathcal{K}$
are norm-preserving, mutually commuting bijections on~$\mathcal{H}$.
$\mathcal{P}$~is linear, while $\mathcal{T}$ and $\mathcal{K}$
are antilinear (\ie~conjugate-linear) operators.
$\mathcal{P}$ and $\mathcal{K}$ are involutive
(\ie~$\mathcal{P}^2=1=\mathcal{K}^2$),
while~$\mathcal{T}$ satisfies~\eqref{antiunitary}.

\begin{Proposition}\label{Prop.symmetry}
$H_0$ is
\begin{itemize}
\item
$\PT$-symmetric if, and only if,
$
  A^-=\mathcal{T}A^+\mathcal{T}
$, \ie,
$$
  A^- =
  \begin{pmatrix}
    -\overline{a_{22}} & \overline{a_{21}}\\
    \overline{a_{12}} & -\overline{a_{11}}
  \end{pmatrix}
  \qquad \mbox{for} \qquad
  A^+ =
  \begin{pmatrix}
    a_{11} & a_{12} \\
    a_{21} & a_{22}
  \end{pmatrix}
  \,;
$$
\item
$\mathcal{PK}$-symmetric if, and only if,
$
  A^-=-\mathcal{K}A^+\mathcal{K}\equiv-\overline{A^+}
$
, \ie,
$$
  A^- =
  \begin{pmatrix}
    -\overline{a_{11}} & -\overline{a_{12}}\\
    -\overline{a_{21}} & -\overline{a_{22}}
  \end{pmatrix}
  \qquad \mbox{for} \qquad
  A^+ =
  \begin{pmatrix}
    a_{11} & a_{12} \\
    a_{21} & a_{22}
  \end{pmatrix}
  \,.
$$
\end{itemize}
\end{Proposition}
\begin{proof}
Since the space $H^2((-a,a);\mathbb{C}^2)$ is left invariant
under the actions of~$\P$, $\T$ and~$\mathcal{K}$,
it is enough to impose algebraic conditions on~$A^\pm$
so that the symmetry properties are ensured. 
More specifically, we need to ensure that 
$\Psi \in \Dom(H_0)$ implies $\mathcal{PT}\Psi \in \Dom(H_0)$.
Employing the identity 
\begin{equation*}
\begin{aligned}
(\mathcal{PT}\Psi)'(\pm a) + A^\pm(\mathcal{PT}\Psi)(\pm a) 
&= (-\mathcal{T}\Psi)'(\mp a) + A^\pm(\mathcal{T}\Psi)(\mp a) \\
&= -\mathcal{T}\left[\Psi'(\mp a) 
+ \mathcal{T} A^\pm \mathcal{T}\Psi(\mp a)\right]
\end{aligned}
\end{equation*}
and the bijectivity of~$\T$, the $\PT$-symmetry condition follows.
The $\mathcal{PK}$-symmetry condition can be established 
in the same manner.
\end{proof}

Another property we would like to examine in this section
is related to the notion of \emph{$S$-self-adjointness}.
We say that a densely defined operator~$H$ on a Hilbert space
is $S$-self-adjoint if
\begin{equation}\label{Jsa}
  H^*=S^{-1} H S
\end{equation}
for some bounded and boundedly invertible
(possibly antilinear) operator~$S$,
where~$H^*$ denotes the adjoint of~$H$.
It clearly generalizes the notion of self-adjointness
and pseudo-Hermiticity.
\begin{Proposition}\label{Prop.sa}
$H_0$ is
\begin{itemize}
\item
self-adjoint if, and only if,
$(A^\pm)^* = A^\pm\,;$
\item
$\P$-self-adjoint if, and only if, 
$A^- = -(A^+)^*$, \ie,
$$
  A^- =
  \begin{pmatrix}
    -\overline{a_{11}} & -\overline{a_{21}}\\
    -\overline{a_{12}} & -\overline{a_{22}}
  \end{pmatrix}
  \qquad \mbox{for} \qquad
  A^+ =
  \begin{pmatrix}
    a_{11} & a_{12} \\
    a_{21} & a_{22}
  \end{pmatrix}
  \,;
$$
\item
$\T$-self-adjoint if, and only if,
$(A^\pm)^* = -\mathcal{T}A^\pm	\mathcal{T}$, \ie,
$$
  A^\pm =
  \begin{pmatrix}
    a^\pm &  0 \\
    0 & a^\pm
  \end{pmatrix}
  \qquad\mbox{with}
  \qquad
  a^\pm \in \mathbb{C}
  \,;
$$
\item
$\mathcal{K}$-self-adjoint if, and only if, 
$(A^\pm)^* = \mathcal{K}A^\pm\mathcal{K} \equiv \overline{A^\pm}\,;$
\end{itemize}
\end{Proposition}
\begin{proof}
The claims follows by using similar arguments 
as in the proof of Proposition~\ref{Prop.symmetry}.
\end{proof}

The spectral analysis of non-self-adjoint operators
is more difficult than in the self-adjoint case,
partly because the \emph{residual spectrum}
is in general not empty for the former.
One of the goals of the present paper is to point out
that the existence of this part of spectrum
is always ruled out for $S$-self-adjoint operators
with antilinear~$S$.
\begin{Proposition}[General fact]
Let~$H$ be a densely defined closed linear operator
on a Hilbert space satisfying~\eqref{Jsa}
with a bounded and boundedly invertible
antilinear operator~$S$.
Then the residual spectrum of~$H$ is empty.
\end{Proposition}
\begin{proof}
Since~$H$ is $S$-self-adjoint, it is easy to see that
$\lambda$~is an eigenvalue of~$H$ (with eigenfunction~$\Psi$)
if, and only if,
$\bar{\lambda}$~is an eigenvalue of~$H^*$
(with eigenfunction~$S^{-1}\Psi$).
It is then clear from the general identity
\begin{equation*}
  \sigma_\mathrm{r}(H) =
  \left\{
  \lambda\in\Com\ | \
  \bar{\lambda}\in\sigma_\mathrm{p}(H^*)
  \ \& \
  \lambda\not\in\sigma_\mathrm{p}(H)
  \right\}
\end{equation*}
that the residual spectrum of~$H$ must be empty.
\end{proof}

The proposition generalizes the fact pointed out in~\cite{BK}
for $S$-self-adjoint operators with~$S$ being a conjugation operator
(\eg~$\mathcal{K}$) and applies to our (different) choice of~$\T$.

\section{Spectral analysis}\label{Sec.spectrum}
%

\subsection{Location of the spectrum and pseudospectrum}
As a consequence of Proposition~\ref{msectorial},
we know that the numerical range of~$H_b$ is contained
in a sector of the complex plane. 
Since the spectrum is a subset of 
the closure of the numerical range,
it provides a basic information on
the location of the spectrum of~$H_b$.
However, coming back to the inequality~\eqref{fundamental}
on which the proof of Lemma~\ref{Lem.relative} is based,
we are able to establish a better result in our case.
\begin{Proposition} \label{parabola.spectrum}
The spectrum of~$H_b$ is enclosed in a parabola, 
\begin{equation*}
\begin{aligned}
\sigma(H_b) \subset \Xi_b := 
\bigg\{ 
z \in \mathbb{C} 
\ \Big| \quad
\Re z &\geq - \left(|b|+ 4\,|A|^2 + \frac{|A|}{2a}\right) =: C
\,,
\\
|\Im z| &\leq \sqrt{8}\,|A|\,\sqrt{\Re z + C} 
+ \frac{|A|}{2a} 
\ \bigg\},
\end{aligned}
\end{equation*}
where $|A| := |A^+| + |A^-|$.
\end{Proposition}
\begin{proof}
By \cite[Corol.~VI.2.3]{Kato},
the numerical range of~$H_b$ is a dense subset 
of the numerical range of its form~$h_b$,
the latter being defined as the set of all complex numbers $h_b[\Psi]$ 
where~$\Psi$ changes over all $\Psi \in \Dom(h_b)$ 
such that $\|\Psi\|=1$. 
Using the first inequality of~\eqref{fundamental}, we get
\begin{equation*}
\begin{aligned}
\Re h_b[\Psi] &\geq q_1[\Psi] + b \, q_2[\Psi] - |q_3[\Psi]| 
\\
&\geq \|\Psi'\|^2 - |b| \, \|\Psi\|^2 
- 2 \, |A| \, \|\Psi'\|\,\|\Psi\|
- \frac{|A|}{2a} \, \|\Psi\|^2 
\\
&\geq \frac{1}{2} \, \|\Psi'\|^2
- \left( |b| + 4 \, |A|^2 + \frac{|A|}{2a}\right) 
\|\Psi\|^2
\,,
\\
|\Im h_b[\Psi]| &\leq |q_3[\Psi]| 
\leq 2 \, |A| \, \|\Psi'\|\,\|\Psi\|
+ \frac{|A|}{2a} \, \|\Psi\|^2 
\,,
\end{aligned}
\end{equation*}
for every $\Psi \in \Dom(h_b)$.
The claim follows by combining these two estimates.
\end{proof}
Thus the resolvent set of~$H_b$ 
contains the complement of $\Xi_b$ in $\mathbb{C}$. 
As a further consequence, we can establish an upper bound 
on the norm of the resolvent:
$$
\|(H_b - z)^{-1}\| \leq 1/{\rm dist}\left(z, \partial \Xi_b\right) 
\qquad \text{for all} \quad z \in \mathbb{C}\setminus\Xi_b.
$$
This result can be also interpreted as 
a location of the \emph{pseudospectrum} of~$H_b$,
\cf~\cite[Sec.~9.3]{Davies_2007}.

\begin{Remark}
Note that the set~$\Xi_b$ in Proposition~\ref{parabola.spectrum}
is not symmetric with respect to the real axis.
On the other hand, if~$H_b$ is $\mathcal{C}$-symmetric
with antiunitary~$\mathcal{C}$ 
(\eg, if~$H_b$ is $\mathcal{PK}$-symmetric),
then we \emph{a priori} know that the numerical range
must be symmetric with respect to the real axis
and an improved version of Proposition~\ref{parabola.spectrum} holds.
\end{Remark}

\subsection{The nature of the spectrum}
Since the Neumann Laplacian~$H_0^N$ has compact resolvent
and the relative bound in Lemma~\ref{Lem.relative}
can be chosen less than~$1/2$ (in fact, arbitrarily small),
it follows from \cite[Thm.~VI.3.4]{Kato}
that~$H_b$ has compact resolvent as well
(for any choice of~$A^\pm$).
\begin{Proposition}
$H_b$ has a purely discrete spectrum
(\ie~any point in the spectrum is an isolated eigenvalue
of finite algebraic multiplicity).
\end{Proposition}

Solving the eigenvalue problem $H_b\Psi=\lambda\Psi$ consists in
constructing the fundamental system of $-\psi_\pm''=k_\pm^2\psi_\pm$
(say, in terms of sines and cosines),
with $k_\pm:=\sqrt{\lambda \mp b}$,
and subject it to the boundary conditions~\eqref{bc.pm}.
This leads to the following algebraic equation
for the eigenvalues~$\lambda$:
\begin{align}\label{eigenvalues}
& \left[\det(A^+) + \det(A^-) - a_{11}^+ a_{22}^- 
- a_{22}^+ a_{11}^-\right] k_- k_+ \cos(a k_-) \cos(a k_+) 
\nonumber\\&
+ \left[\det(A^+) \det(A^-) +a_{11}^+ a_{11}^- k_-^2 
+a_{22}^+ a_{22}^- k_+^2 +k_-^2 k_+^2\right]\sin(a k_-) \sin(a k_+)
\nonumber\\&
+ \left[-\det(A^+) a_{22}^- + a_{22}^+ \det(A^-) 
+(-a_{11}^+ + a_{11}^-) k_-^2\right] k_+ \sin(a k_-) \cos(a k_+)
\nonumber\\&
+ \left[-\det(A^+) a_{11}^- + a_{11}^+ \det(A^-) 
+ (- a_{22}^+ +a_{22}^-) k_+^2\right] k_- \cos(a k_-) \sin(a k_+)
\nonumber\\&
+\left(a_{21}^+ a_{12}^- + a_{12}^+ a_{21}^-\right) k_- k_+ = 0,
\end{align}
where $a_{ij}^+$ and $a_{ij}^-$ denote the elements 
of the matrices $A^+$ and $A^-$, respectively.

There are only a few choices of~$A^\pm$ 
for which~\eqref{eigenvalues} admits explicit solutions. 
In the sequel we consider some particular situations
that we analyse with help of numerical solutions.

\subsection{Examples}
\paragraph{A self-adjoint example with avoided crossings.}
Let us choose
\begin{equation}\label{Ex.a}
A^\pm :=
\begin{pmatrix}
0 & \ii \alpha\\
-\ii \alpha & 0
\end{pmatrix}
  \,,
\end{equation}
where~$\alpha$ is a real parameter.
It follows from from Proposition~\ref{Prop.sa}
that all the eigenvalues are real since~$H_b$ is self-adjoint.
The implicit equation for the eigenvalues takes form
\begin{equation*}
  2 \alpha^2 k_+ k_- [1 - \cos (2a k_+)\cos (2a k_-)]
  = - (k_+^2 k_-^2 + \alpha^4)\sin(2a k_+)\sin(2a k_-).
\end{equation*}
The dependence of eigenvalues on the parameter~$\alpha$
can be seen in Figure~\ref{exampleB}.
An interesting phenomenon in this figure is the approaching
of a pair of eigenvalues and its subsequent moving back
and slowly approaching to constant values.
It should be noted that in the point of closest approach
the two curves do not intersect.
This avoided crossing holds for each pair of the eigenvalues.

\begin{figure}[h!]
\begin{center}
\begin{tabular}{cc}
\includegraphics[width=0.45\textwidth]{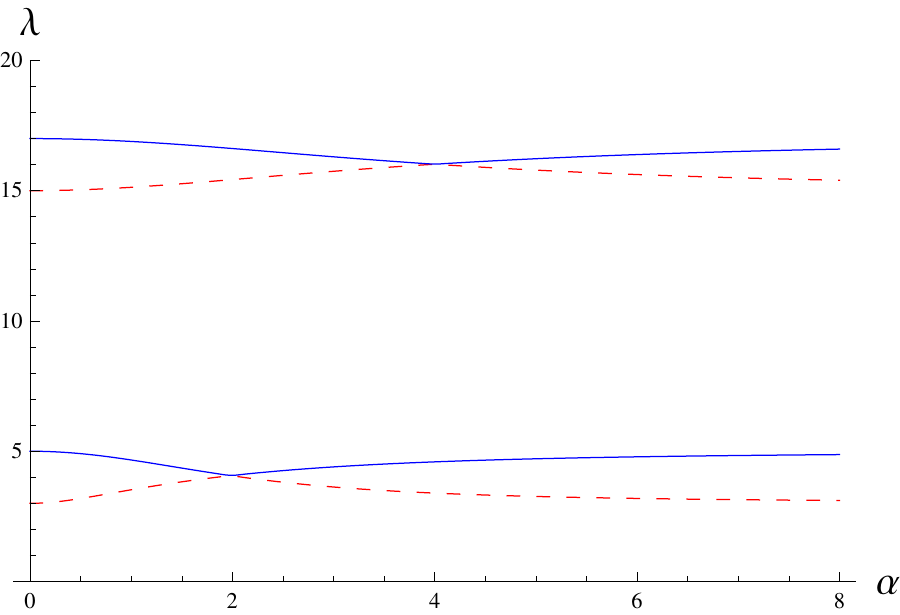}
& \includegraphics[width=0.45\textwidth]{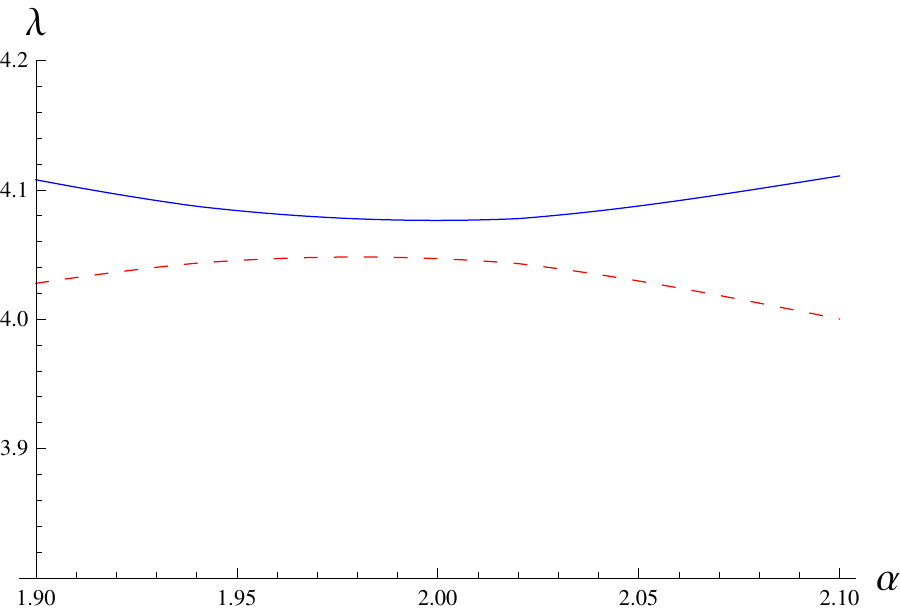}
\end{tabular}
\end{center}
%
\caption{$\alpha$-dependence of eigenvalues
for $b = 1$ and $a = \frac{\pi}{4}$ in example~\eqref{Ex.a},
with a zoom of the avoided crossing
of the first pair of eigenvalues on the right.}
\label{exampleB}
\end{figure}

\paragraph{A $\mathcal{PT}$-symmetric example
with real and complex spectra.}
As an example of non-Hermitian but $\mathcal{PT}$-symmetric
boundary conditions, let us consider
\begin{equation}\label{Ex.b}
A^\pm =
\begin{pmatrix}
\ii \alpha \pm \beta & 0\\
0 & \ii \alpha \pm \beta
\end{pmatrix}
  \,,
\end{equation}
where~$\alpha$ and~$\beta$ are real parameters.
The feature of this example is that the spinorial
components do not mix.	
The implicit equation for the eigenvalues acquires the form
\begin{multline}\label{eigenC}
 \left(-2 \beta k_- \cos(2 a k_-)
 + (k_-^2 - \alpha^2 - \beta^2) \sin(2 a k_-)\right)
 \\
 \times \left(-2 \beta k_+ \cos(2 a k_+)
 + (k_+^2 - \alpha^2 - \beta^2) \sin(2 a k_+)\right) = 0.
\end{multline}

Because of the decoupling, this eigenvalue problem can be analysed
by using known results for this type of boundary conditions
in the scalar case previously studied in~\cite{KBZ}
and in more detail in~\cite{KS}.
It turns out that the spectrum significantly
depends on the sign of~$\beta$.

\medskip
\noindent
\underline{$\beta = 0$}.
It follows from \cite{KBZ} that one pair of eigenvalues
depend on the parameter~$\alpha$ quadratically
and the others are constant, 
see the left part of Figure~\ref{exampleCb}.
More specifically, the eigenvalues explicitly read
\begin{equation}\label{eig}
\lambda_{j,\pm}=
\begin{cases}
  \alpha^2 \mp b
  & \mbox{if} \quad j=0
  \,,
  \\
  \left(\frac{\displaystyle j \pi}{\displaystyle 2a}\right)^2 \mp b
  & \mbox{if} \quad j \geq 1
  \,.
\end{cases}
\end{equation}
The crossings of full (respectively dashed)
lines in the left part of Figure~\ref{exampleCb}
correspond to eigenvalues of geometric multiplicity one
and algebraic multiplicity two,
while the crossings of full lines with dashed lines
correspond to eigenvalues of both multiplicities equal to two.
The entire spectrum is doubly degenerate for~$b=0$
and there exist eigenvalues of geometric multiplicity two
and algebraic multiplicity four.

%
%
%
%

\medskip
\noindent
\underline{$\beta > 0$}.
In this case, the reality of the spectrum was proved in~\cite{KS}.
The right part of Figure~\ref{exampleCb} 
shows the dependence of the eigenvalues
on the parameter~$\alpha$.
We again observe pairs of eigenvalues split 
because of the presence of the magnetic field.

\begin{figure}
\begin{center}
\begin{tabular}{cc}
\includegraphics[width=0.45\textwidth]{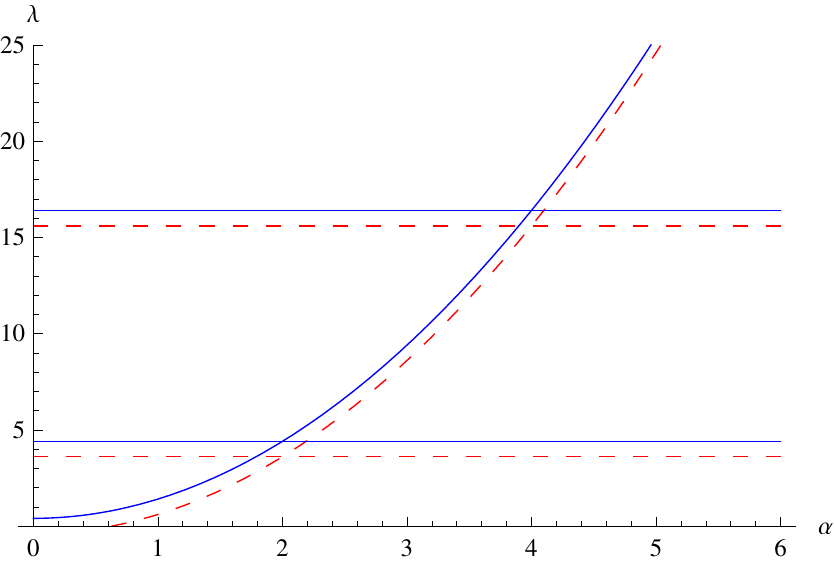}
& \includegraphics[width=0.45\textwidth]{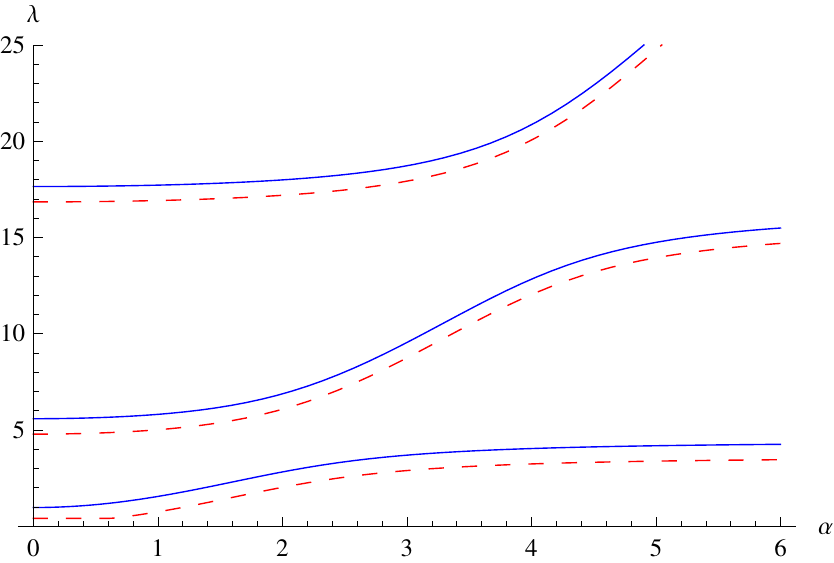}
\\
$\beta=0$
& $\beta=0.5$
\end{tabular}
\end{center}
%
\caption[ExampleA]{$\alpha$-dependence of eigenvalues
for $b = 0.4$ and $a = \frac{\pi}{4}$ in example~\eqref{Ex.b}.}
\label{exampleCb}
\end{figure}

\medskip
\noindent
\underline{$\beta < 0$}.
On the other hand, the reality of the spectrum
in the case when~$\beta$ is negative
is not guaranteed and, indeed, it is easily seen
from Figure~\ref{exampleCc} that complex conjugate pairs
of eigenvalues do appear when a couple of
real eigenvalues collides as enlarging~$\alpha$.
The pair of complex eigenvalues becomes real again
for larger values of~$\alpha$.
It follows from the analysis in~\cite{KS} that
only one pair of complex conjugate eigenvalues
occurs simultaneously in the spectrum.

\begin{figure}[h!t]
\begin{center}
\begin{tabular}{cc}
\includegraphics[width=0.45\textwidth]{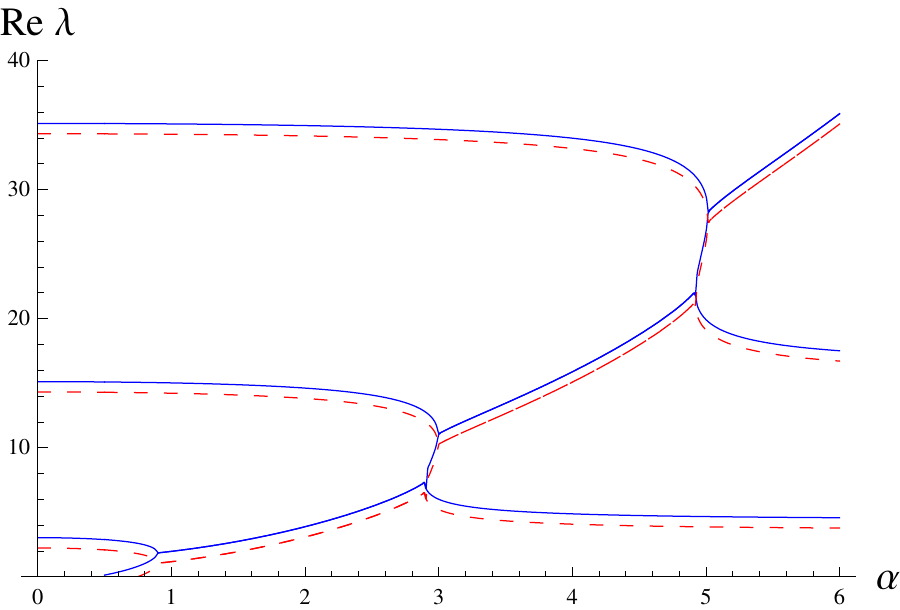}
&\includegraphics[width=0.45\textwidth]{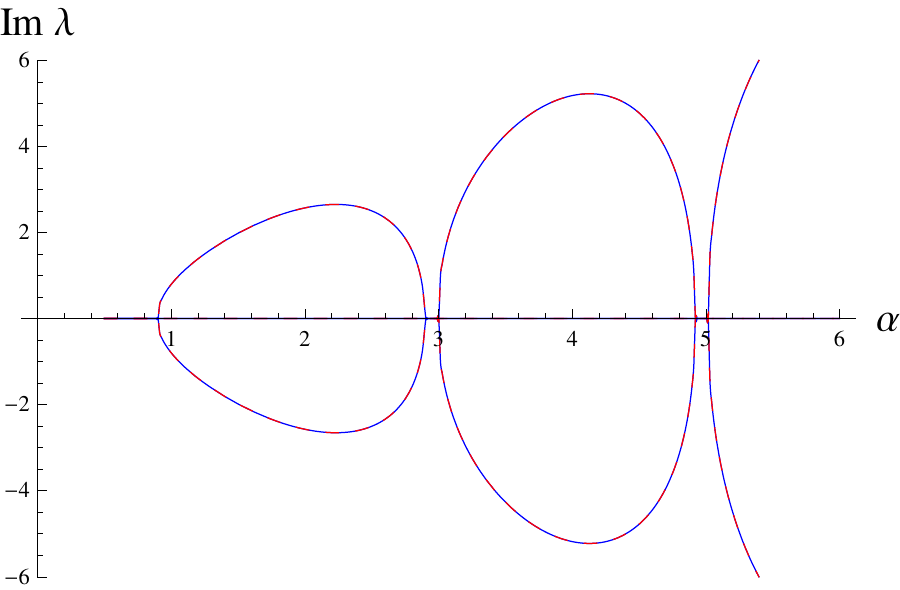}
\end{tabular}
\end{center}
%
\caption[ExampleC]{$\alpha$-dependence of 
the real (left) and imaginary (right) parts of eigenvalues
for $b = 0.4$, $a = \frac{\pi}{4}$ and $\beta=-0.5$
in example~\eqref{Ex.b}.}
\label{exampleCc}
\end{figure}

\paragraph{A $\PT$-symmetric example with coupled spinorial components}
As another example of non-Hermitian 
$\PT$-symmetric boundary conditions,
let us select
\begin{equation}\label{PTexample}
A^{\pm} =
\begin{pmatrix}
0 & \pm \ii \alpha\\
\pm \ii \alpha & 0
\end{pmatrix},
\end{equation}
where~$\alpha$ is a real parameter.
The characteristic feature of this model is a non-trivial 
mixing of spinorial components. 
The implicit equation for the eigenvalues now takes the form
\begin{equation}
\begin{aligned}
 4 \alpha^2 k_+ k_- \cos(a k_+)^2 \cos(a k_-)^2&+4 \alpha^2 k_+ k_- \sin(a k_+)^2 \sin(a k_-)^2 \\
&= -(k_+ k_- +\alpha^4) \sin(2 a k_+) \sin(2 a k_-).
\end{aligned}
\end{equation}

The dependence of low-lying eigenvalues 
on the parameter~$\alpha$ can be seen in Figure~\ref{exampleD}.
Here the lowest pair of real eigenvalues exhibits a crossing,
however, the eigenvalues remain real after the crossing point
as the parameter~$\alpha$ increases. 
This behaviour is not featured uniquely
by the lowest pair of eigenvalues, 
it also appears for higher-lying eigenvalues in the spectrum
(not visible in the figure).
On the other hand, as~$\alpha$ increases, 
the other pairs of eigenvalues in the figure
complexify after the first collision,
then the corresponding eigenvalues 
propagate as complex conjugate pairs in the complex plane,
meet again and become real.

\begin{figure}[h!]
\begin{center}
\begin{tabular}{cc}
\includegraphics[width=0.45\textwidth]{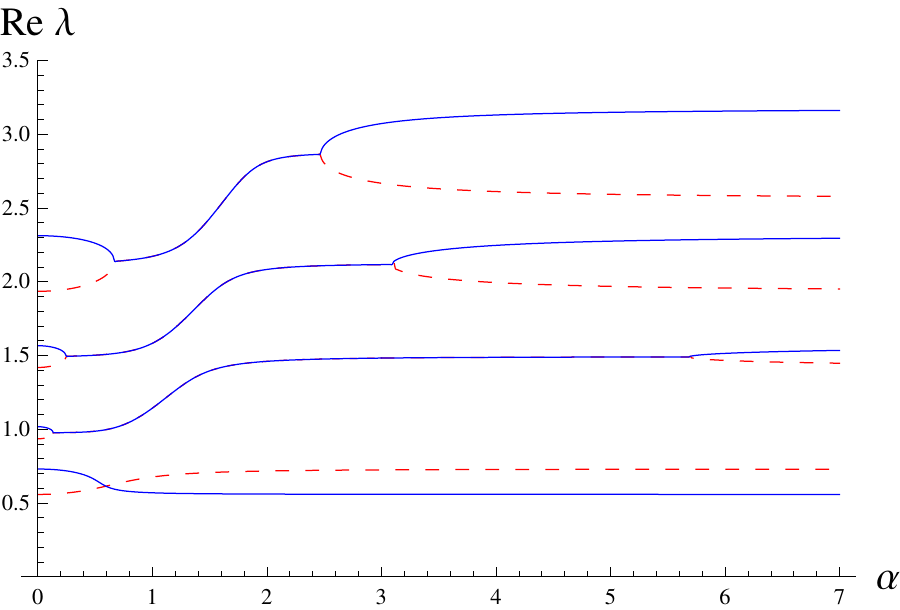}
&\includegraphics[width=0.45\textwidth]{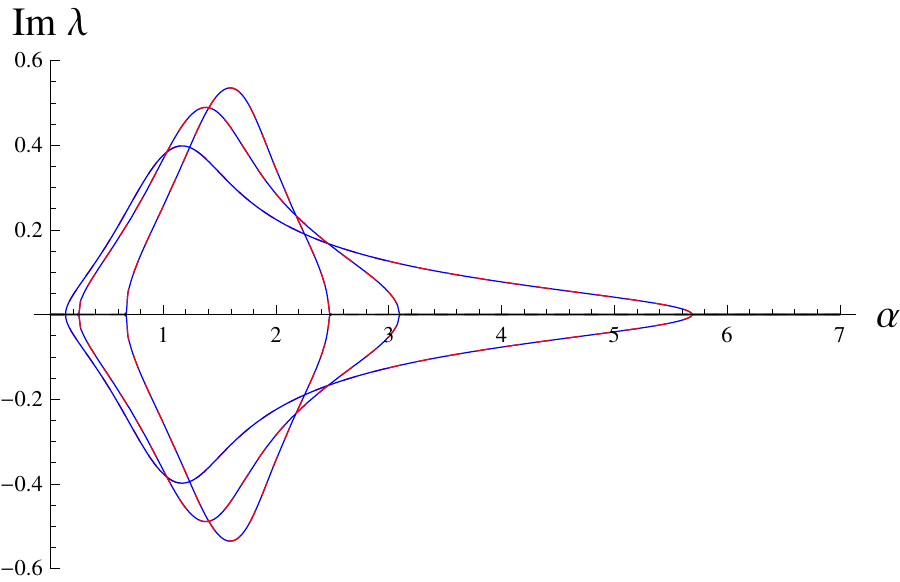}
\end{tabular}
\end{center}
\caption{$\alpha$-dependence of eigenvalues 
for $b = 0.5$, $a = \sqrt{43}$ in example~\eqref{PTexample}.
An animation can be found on the website~\cite{KKNS-video}.}
\label{exampleD}
\end{figure}
%

\section{Conclusions}\label{Sec.end}
%
The goal of this paper was to investigate the role of spin
in complex extensions of quantum mechanics
on a simple model of Pauli equation
with complex Robin-type boundary conditions.
A special attention was paid to $\PT$-symmetric situations
with a physical choice of the time-reversal operator~$\mathcal{T}$.

A simple physical interpretation of our model in terms 
of scattering was suggested in Section~\ref{Sec.scattering}.  
It would be desirable to examine this motivation in more details
and include ``spin-dependent electric potential''
(\eg\ Bychkov-Rashba or Dreselhauss spin-orbit terms
typical for semiconductor physics~\cite{spintronics}). 

Robin boundary conditions represent a class
of separated boundary conditions.
Our model can be naturally extended to
connected boundary conditions,
whose spectral analysis represents a direction
of potential future research
(\cf~\cite{KS,Ergun_2010,Ergun-Saglam_2010} in the scalar case).

In this paper we did not discuss the important question
of the existence of similarity transformations
(or the ``metric'' in the $\PT$-symmetric context)
connecting our non-Hermitian operators with
self-adjoint Hamiltonians.
The problem generally constitutes a difficult task
and very few closed formulae are known
(\cf~\cite{K4,Assis-Fring_2009,Assis_2011,KSZ} and references therein).
However, we can easily extend the results established
in the scalar case without magnetic field~\cite{KSZ} 
to our spinorial example~\eqref{Ex.b}
and compute the metric in this special case. 
Let us define
\begin{equation*}
\Theta
:= \begin{pmatrix} I+K & 0 \\ 0 & I+K \end{pmatrix}
\,,
\end{equation*}
where~$I$ denotes the identity operator on $\sii((-a,a))$
and~$K$ is an integral operator with kernel
\begin{equation*}
K(x,y):=e^{\ii \alpha (x-y) - \beta |x-y| } \,
\big[c+\ii \alpha  \sgn(x-y) \big]
\,, 
\end{equation*}
with~$c$ being any real number.
It follows from~\cite[Sec.~4.5]{KSZ} 
and the nature of the decoupled boundary conditions~\eqref{Ex.b} 
that~$\Theta$ represents a one-parametric family of metrics for~$H_b$
under the $\PT$-symmetric choice~\eqref{Ex.b}.
More precisely, $H_b$~is $\Theta^{-1}$-self-adjoint
(\cf~\eqref{Jsa})
and~$\Theta$ is \emph{positive} provided that either:
$a$~is small;
or $\beta$ is positive and large;
or $|c|$ and $|\alpha|$ are small.
To find the self-adjoint counterpart of~$H_b$
determined by this similarity transformation
constitutes an open problem 
(in the scalar case~\cite{KSZ} 
there exists results for~$\beta=0$).	 

Our model was effectively one-dimensional.
Higher dimensional generalizations in the spirit
of \cite{BK,BK2,Olendski_2011} would be especially
interesting for variable boundary conditions
(\ie~non-constant matrix~$A$).

\subsection*{Acknowledgement}
The last three authors acknowledge the hospitality 
of the Comenius University in Bratislava
where this work was initiated;
these authors have been partially supported 
by the GACR grant No.\ P203/11/0701.
The first author acknowledges support from DFG SFB-689.
The last author has been supported by a grant within 
the scope of FCT's project PTDC/MAT/101007/2008 
and partially supported by FCT's projects 
PTDC/MAT/101007/2008 and PEst-OE/\ MAT/UI0208/2011.

%
{\small
\bibliography{bib}
\bibliographystyle{amsplain}
}
\end{document}